\documentclass[11pt, oneside]{amsart}
\usepackage{geometry}                
\usepackage{color}
\usepackage[colorlinks=true]{hyperref}
\usepackage{graphicx}
\usepackage{amssymb}
\usepackage{epstopdf}
\usepackage[english]{babel}
\usepackage{tabls}
\usepackage{slashed}
\usepackage{tipa}
\usepackage{color}
\usepackage[all]{xy}
\usepackage{eucal}

\newcommand{\eps}{\epsilon}

\newcommand{\R}{\mathbb{R}}

\newcommand{\ud}{\text{d}}

\def\sdiv{\text{\texthtd}}
\def\scurl{\text{\texthtc}}
\def\stwist{\text{\texthtt}}

\newcommand{\coder}{\ud^*}

\newtheorem{theorem}{Theorem}[section]
\newtheorem{remark}[theorem]{Remark}
\newtheorem{lemma}[theorem]{Lemma}

\newtheorem{definition}[theorem]{Definition}
\newtheorem{proposition}[theorem]{Proposition}

\DeclareMathOperator*{\diverg}{div}

\def\clap#1{\hbox to 0pt{\hss#1\hss}}
\def\mathclap{\mathpalette\mathclapinternal}
\def\mathclapinternal#1#2{\clap{$\mathsurround=0pt#1{#2}$}}

\newcommand\blfootnote[1]{%
  \begingroup
  \renewcommand\thefootnote{}\footnote{#1}%
  \addtocounter{footnote}{-1}%
  \endgroup
}

\usepackage[backend=bibtex, style=numeric, isbn=false, url=false]{biblatex}
\bibliography{biblio1}

\title[]{Gluing for the constraints for higher spin fields}
\author[J. Joudioux]{J\'er\'emie Joudioux}
\email{jeremie.joudioux@univie.ac.at}
\address{Gravitationsphysik, Fakult\"at f\"ur Physik, Universit\"at Wien, Boltzmanngasse 5, 1050 Wien, Austria}

\begin{document}
\blfootnote{Preprint Number: UwThPh-2017-3-7}
\maketitle
\begin{abstract} This short note is a follow-up on the paper by Beig and Chru\'sciel \cite{2017arXiv170100486B} regarding the use of potentials to perform a gluing and shielding of initial data for Maxwell fields and linearised gravity. Based on a work in collaboration with Andersson and B\"ackdahl \cite{MR3238530}, this gluing and shielding procedure is generalised to higher spin fields. The approach is based on a generalisation of the de Rham complex to higher spin fields providing a parametrization of the set of constraints, as well as standard elliptic theory to prove the existence of a potential.
\end{abstract}

\section*{Introduction}

In a recent paper  \cite{2017arXiv170100486B}, Beig and Chru\'sciel produced an elementary way to glue, and \emph{shield}, solutions to the linearised constraint equation by relying on the representation of solutions by a potential. Their elementary approach is motivated by the recent work of Carlotto and Schoen \cite{MR3539922} who performed a gluing and \emph{shielding} of initial data for the Einstein equations, whose Cauchy development is non interacting for an arbitrary long time. This \emph{shielding} of the gravitational effect is a relativistic phenomenon, in the sense that it does not occur in Newtonian gravity. This raises the question of whether this form of shielding holds for linearised gravity, and, by extension, to arbitrary spin fields.

The approach developed by Beig and Chru\'sciel is based on the remark that the solution to the constraint equations to linearised gravity can be integrated explicitly, as described by and Beig \cite{Beig:1997wo} (see also Gasqui-Goldschmidt \cite{2017arXiv170100486B}). Once a potential to a field is constructed, the gluing procedure is elementary.

The construction of Hertz potentials (satisfying a wave equation) for zero rest-mass fields has been introduced by Penrose to study the asymptotic behaviour of the null components of the field in the outgoing null direction (the \emph{peeling}). This approach to the asymptotic behaviour of zero rest-mass fields has been extended in the analytic context of the Cauchy problem by Andersson, B\"ackdahl and Joudioux \cite{MR3238530}. In particular, the relation between the initial datum of the zero rest-mass fields, satisfying the higher spin equivalent of geometric constraints for the Maxwell fields (spin one) and linearised gravity (spin two), and the initial data of the Hertz potential is described there in detail. This relation, obtained by a $3+1$-splitting of the representation by a Hertz potential, is described by an elliptic complex, which generalizes the de Rham complex (as well as the Gasqui-Goldschmidt-Beig complex \cite{Gasqui:1984vu,Beig:1997wo}) to higher spin fields in dimension $3$.

It is well-known that the de Rham complex describes the integrability property of the constraint equations for the Maxwell fields. In a similar fashion, the extension of the de Rham complex to higher spin fields described in \cite{MR3238530} provides an algebraic description of the geometric constraints for these higher spin fields. An important and interesting consequence, which we would like to emphasize here, of the approach of \cite{MR3238530} is that the set of solutions to the constraint equations, on $\mathbb{R}^3$, can be parametrized: more specifically, there exists a differential operator $\mathcal {G}$, between spinor fields of spin $s$, such that, if $\psi$ is any spinor field of spin $s$, then $\phi = \mathcal{G}\psi$ satisfies the constraint equation for spin-$s$ fields:
$$
D^{AB} \phi_{A \dots F} = 0,
$$
where $D$ is the standard connection on the Euclidean space $\mathbb{R}^3$. Using standard elliptic theory, the surjectivity of the operator $\mathcal{G}$, between appropriate weighted Sobolev spaces, is proven in \cite[Sections 3 and 4]{MR3238530}. An immediate by-product of this parametrization of the set of constraints is the possibility to perform, in an elementary way, gluing of solutions, as observed in \cite{2017arXiv170100486B} for the spin 1 and spin 2 cases.

More specifically, consider a spinor field of spin $s\geq 1$, $\phi_{A\dots F}$, satisfying the constraint equation for zero rest-mass fields,
$$
D^{AB} \phi_{A\dots F} = 0.
$$
Given $\Omega$ an open domain; consider an $\epsilon-$neighbourhood $\tilde{\Omega}_{\epsilon}$ of $\Omega$ obtained by $\cup_{x\in \Omega}B(x, \epsilon <x>)$ (with $<x> = (1+x^2)^{1/2}$). One proves that there exists a spinor field $\tilde \phi_{A\dots F}$, satisfying the the constraint equation
$$
D^{AB} \tilde{\phi}_{A\dots F} = 0
$$
agreeing with $\phi$ on $\Omega$, and $0$ outside a bigger set $\tilde{\tilde{\Omega}}_{\epsilon}$. Furthermore, if $\Omega$ is unbounded, and if $\phi$, regular enough, satisfies, for $\delta \in \mathbb{R}\setminus \mathbb{Z}$,
$$
\phi(x) =  \mathcal{O}(|x|^\delta) \text{ as }|x| \rightarrow +\infty
$$
then, in the overlapping region $\tilde{\tilde{\Omega}}_{\epsilon}\setminus \Omega$, $\tilde{\phi}$ satisfies
$$
\tilde{\phi}(x) =  \mathcal{O}(|x|^\delta) \text{ as }|x| \rightarrow +\infty.
$$

The paper is organized as follows. The first section  \label{sec:prel} contains basic preliminaries on spinors and weighted Sobolev spaces. The second section \label{sec:hertz} introduces the Hertz potentials. The third section \label{sec:complex} is a digest of \cite[Sections 3 and 4]{MR3238530}, and explains the elliptic complex for higher spin fields and relates it to the geometric constraint of zero rest-mass fields. The last section \label{sec:glue} contains the gluing result. The appendix contains the mechanism to construct a potential.

\section{Preliminaries}\label{sec:prel}

\subsection{Geometric context}
In this paper, the main focus is on $\mathbb{R}^3$, endowed with the Euclidean metric of \emph{negative signature}.  Nonetheless, since this problem comes from the study of the Cauchy problem for zero rest-mass fields, we specify, when appropriate, when we are working on $\mathbb{R}^4$. The situation is normally clear from context.

All along this paper, the Penrose notations for spinors (see \cite{Penrose:1986fk}) are heavily used, and we assume that the reader has some familiarity with it. Hence, when working on Minkowski space-time, we work the Minkowski metric of signature $(+,-,-,-)$.

We consider $\nabla$ the standard connection on $\mathbb{R}^4$, and $D$ the associated Sen connection on $\mathbb{R}^3$ given by
$$
\nabla_{AA'}=\tfrac{1}{\sqrt{2}}\tau_{AA'}\partial_t-\tau^{B}{}_{A'}D_{AB},
$$
where $\tau_{AA'}=\sqrt{2}\nabla_{AA'}t$.

\subsection{Introduction to space spinors}
We recall here the basics of space spinors. The reader who wishes to have a deeper overview on space spinors can refer to \cite{Sommers:1980dd} and \cite[Chapter 4]{MR3585918}. Let $p_{ab}$ be a 2-tensor; $g$ admits the space-time spinor $g_{AA'BB'}$ representation:
$$
p_{ab} =\psi_{AA'BB'}.
$$

Consider now the unit future oriented timelike vector $n^{AA'}$. This vector can be used as an isomorphism between primed and unprimed indices \footnote{The $\sqrt{2}$-factor is a mere normalization.}:
\begin{align}
\alpha_{A} &\mapsto \sqrt{2} n^{A}{}_{A'} \alpha_{A} \label{eq:iso1}\\
\alpha_{A'} &\mapsto \sqrt{2}n_{A}{}^{A'} \alpha_{A'} \label{eq:iso2}
\end{align}
This procedure can then be used to associate to any space-time spinor (with mixed prime and unprimed indices) a space spinor by means of the of the isomorphism \eqref{eq:iso2}. In the particular case of the tensor $g_{ab}$, the associated space spinor is
$$
\psi_{ABCD} = 2n_{A}{}^{A'} n_{B}{}^{B'}   g_{CA'DB'}.
$$
If the tensor $h_{ab}$ is a four dimensional spinor carrying a non trivial components in the direction $n$, the order of the indices matters. If the tensor is a spatial tensor, that is to say a 2-tensor of $\mathbb{R}^3$, the order of the indices does not matter.

\subsection{Weighted Sobolev spaces} We denote by $S_{2s}$ the bundle of space spinors over $\mathbb{R}^3$. The standard pointwise norm is denoted by $|\cdot|$.

Let $\delta$ in $\mathbb{R}\setminus \mathbb{Z}$. The completion of the space of smooth spinor fields in $S_{2s}$ with compact support in $\R^3$ endowed with the norm
$$
\Vert \phi_{A\dots F}\Vert^2_{j, \delta} = \sum_{n=0}^j \left\Vert <r>^{-(\delta+\frac32) + n} D^n\phi_{A\dots F}\right\Vert^2_{2},
$$
is denoted by $H^j_{\delta}(S_{2s})$, where $<r> = (1+r^2)^{1/2}$. These spaces are the generalization of the standard weighted Sobolev spaces to spinor sections.

Since we are working on $\mathbb{R}^3$, the standard properties associated with these spaces (Sobolev embeddings, Fredholm properties associated with elliptic operators, as described, for instance in \cite{Bartnik:1986dq}).

\section{Representations of zero rest-mass fields by Hertz potentials} \label{sec:hertz}

In \cite{MR3238530}, we addressed the question of constructing a Hertz potential, satisfying the wave equation, for zero-rest-mass fields in the context of the Cauchy problem. This construction was performed as follows:
\begin{itemize}
\item we performed a $3+1$-splitting of the Hertz potential equation; the resulting equation relates the initial datum of the zero-rest-mass field to the initial data of the Hertz potential;
\item this equation is integrated by means of an ad hoc elliptic complex, generalizing to higher spin fields the de Rham complex;
\item finally, the uniqueness of solutions to the Cauchy problem guarantees the representation formula.
\end{itemize}

More precisely, consider a zero-rest-mass spin-$s$ field $\phi_{A \dots F}$, \emph{\emph{i.e.}} a symmetric valence $2s$ spinor field on Minkowski space, which solves
\begin{equation}\label{cauchyprobmassless}
 \left\{
\begin{array}{l}
\nabla^{AA'}\phi_{A\dots F}=0, \\
\phi_{A\dots F}|_{t=0}=\phi_{A\dots F} \in H^{j}_{\delta}(S_{2s}).
\end{array}\right.
\end{equation}
For $s\geq 1$, this Cauchy problem is consistent only when the geometric constraint
\begin{equation}\label{geometricconstraint}
D^{AB} \phi_{ABC\dots F} = 0
\end{equation}
is satisfied.

A \emph{Hertz potential} for the massless spin-s field $\phi_{A \dots F}$ is a symmetric spinor field $\chi ^{A'\dots F'}$ satisfying the wave equation
\begin{equation}
  \square \chi ^{A'\dots F'} =  \eta^{ab}\nabla_a \nabla_b \chi ^{A'\dots F'} = 0
\end{equation}
such that
\begin{equation}
  \phi_{A \dots F} =\nabla_{AA'} \dots \nabla_{FF'} \chi ^{A'\dots F'}.
\end{equation}

The Hertz potentials in the context of zero rest-mass fields were introduced by Penrose \cite{MR0175590} to study the peeling properties of these fields, that is to say the asymptotic behaviour of null components in outgoing null directions. In \cite{MR3238530}, this approach was further developed in the analytic context of the Cauchy problem for higher spin fields with datum in weighted Sobolev spaces. This led to a detailed analysis of the asymptotic behaviour of fields, and, in particular, to a detailed description of the asymptotic behaviour of the initial data ensuring that the peeling property is fulfilled. The relation between the initial data of the field and Hertz potentials is \cite[Sections 3 and 4]{MR3238530} described in a precise fashion, as explained in the next section.

\section{Generalization of the de Rham complex to arbitrary spin} \label{sec:complex}
This section is essentially a digest of \cite[Sections 3 and 4]{MR3238530}. The aim is to generalize the de Rham complex to arbitrary spin, and describe, in terms of operators, the integrability condition on the Euclidean space $\mathbb{R}^3$  of the constraints for the massless higher spin fields:
$$
D^{AA'}\phi_{A\dots F} = 0.
$$
In the case of spin $1$, that is to say for the Maxwell equations, these constraints are
$$
 d F =0, \quad d^\star F = 0, \text{ on }\mathbb{R}^4
$$
that is to say the standard constraint equations for the electric and magnetic parts, $E$ and $B$, of the Faraday tensor:
 $$
d^\star E = d^\star B =  0 \text{ on }\mathbb{R}^3.
 $$
In a similar fashion, the spin-2 field equation \begin{equation}
  \nabla^a W_{abcd} = 0, \text{ on }\mathbb{R}^4
\end{equation}
where $W$ satisfies the symmetry of the Riemann tensor, admits constraint equations for the electric and magnetic part, $E$ and $B$, of $W$:
\begin{equation}
D^aE_{ab} = D^aB_{ab} = 0.
\end{equation}
These algebraic properties of the Maxwell fields and spin 2 fields are well described in \cite{Christodoulou:1990dd}.

For convenience, we introduce the following notations:
\begin{definition} \label{def:fundop}
Let $\phi_{A_1\dots A_k}$ be a spinor field of valence $k$.
Let $D_{AB}$ be the intrinsic Levi-Civita connection on $\mathbb{R}^3$. Define the operators
\begin{align*}
(\sdiv_k\phi)_{A_1\dots A_{k-2}}\equiv{}&D^{A_{k-1}A_{k}}\phi_{A_1\dots A_k},\\
(\scurl_k\phi)_{A_1\dots A_{k}}\equiv{}&D_{(A_1}{}^{B}\phi_{A_2\dots A_{k})B},\\
(\stwist_k\phi)_{A_1\dots A_{k+2}}\equiv{}&D_{(A_1 A_2}\phi_{A_3\dots A_{k+2})}.
\end{align*}
These operators are called \emph{divergence}, \emph{curl} and \emph{twistor operator} respectively.
\end{definition}
For a discussion on these operators, and their use, the reader can refer to \cite[Section 2.1]{MR3231172}.

The constraint equations \eqref{geometricconstraint} for a spin s field are then written
\begin{equation}
  \sdiv_{2s} \phi = 0
\end{equation}

We now define the operator which can be used to parametrize the constraint equations (see \cite[Definition 2.11]{MR3238530}):
\begin{definition}[Andersson, B\"ackdahl, Joudioux, \cite{MR3238530}]
Define the operators, of order $2s-1$: $\mathcal{G}_k:\mathcal{S}_k\rightarrow \mathcal{S}_{k}$ by
\begin{align*}
(\mathcal{G}_k\phi)_{A_1\dots A_k} &\equiv \sum_{n=0}^{\mathclap{\left\lfloor\tfrac{k-1}{2}\right\rfloor}}\binom{k}{2n+1}(-2)^{-n} \underbrace{D_{(A_1}{}^{B_1}\cdots D_{A_{k-2n-1}}{}^{B_{k-2n-1}}}_{k-2n-1} (\Delta^n_{k}\phi)_{A_{k-2n}\dots A_k)B_1\dots B_{k-2n-1}},
\end{align*}
where $\Delta_{k}$ is the Laplacian on spinors of valence $2k$.
\end{definition}
Particular cases (spin one and two) are discussed later.

It is now possible to state the elliptic complex describing the integrability properties of the constraints equation \eqref{geometricconstraint} (see \cite[Lemma 3.3]{MR3238530}): we remind here that a complex of differential operators, that is to say a sequence of differential operators, is elliptic, if the associated sequence made out of the symbols of these operators is exact.
\begin{lemma}[Andersson, B\"ackdahl, Joudioux, \cite{MR3238530}]\label{lem:ellipticcomplex}
The sequence
\begin{equation*} 
\mathcal{S}_{2s-2}\stackrel{\stwist_{2s-2}}{\longrightarrow} \mathcal{S}_{2s} \stackrel{\mathcal{G}_{2s}}{\longrightarrow} \mathcal{S}_{2s}\stackrel{\sdiv_{2s}}{\longrightarrow} \mathcal{S}_{2s-2},
\end{equation*}
is an elliptic complex.
\end{lemma}

In the particular case of the spin 1, this elliptic complex is
\begin{equation}\label{exactseq1}
 \mathcal{S}_{0}\stackrel{\stwist_0}{\longrightarrow} \mathcal{S}_{2} \stackrel{\scurl_2}{\longrightarrow} \mathcal{S}_{2}\stackrel{\sdiv_2}{\longrightarrow} \mathcal{S}_{0}.
\end{equation}
and is the spinorial equivalent of the variation of the well-known de Rham complex
\begin{equation*}
 C^\infty(\R^3,\R)\stackrel{\ud}{\longrightarrow} \Lambda^1\stackrel{\star \ud}{\longrightarrow} \Lambda^1  \stackrel{\coder}{\longrightarrow}C^\infty(\R^3,\R).
\end{equation*}
In the spin-2 case, it generalizes the Gasqui-Goldschmidt-Beig complex (see \cite[Theorem 6.1  (2.24)]{Gasqui:1984vu} and \cite{Beig:1997wo}), stated here in the case of $\mathbb{R}^3$:
$$
\Lambda^1(\mathbb{R}^3) \stackrel{L}{\longrightarrow}  S_0^2(\mathbb{R}^3, g) \stackrel{\mathcal{R}}{\longrightarrow} S_0^2(\mathbb{R}^3, g) \stackrel{\diverg} {\longrightarrow}\Lambda^1(\mathbb{R}^3),
$$
where $\Lambda^1(\mathbb{R}^3)$ is the space of 1-forms over $\mathbb{R}^3$, $S_0^2(\mathbb{R}^3, g)$ is the space of symmetric trace free 2-tensors and
\begin{eqnarray*}
(LW)_{ab}&=&D_{(a}W_{b)}-\frac{1}{3}g_{ab}D^cW_c\\
(\diverg t)_a&=&2g^{bc}D_ct_{ab}
\end{eqnarray*}
and \begin{equation*}
\begin{array}{c}
\mathcal{R}(\psi)_{ab} = \eps^{cd}{}_a D_{[c}\sigma_{d]b} \text{ where}\\
\sigma_{ab}=D_{(a} D^c \psi_{b)c} - \frac{1}{2} \Delta \psi_{ab} - \frac{1}{4} g_{ab}D^c D^d \psi_{cd}.
\end{array}
\end{equation*}
In terms of spinors, the spin 2 complex is
\begin{equation}\label{exactseq2}
\mathcal{S}_2 \stackrel{\stwist_2}{\longrightarrow}  \mathcal{S}_4 \stackrel{\mathcal{G}_4}{\longrightarrow} \mathcal{S}_4 \stackrel{\sdiv_4} {\longrightarrow}\mathcal{S}_2.
\end{equation}
where the operator $\mathcal{G}_4$ is related to the linearised Cotton-York tensor $\mathcal{R}$ by the relation:
$$
\mathcal{R}_{ab} = \mathcal{R}_{ABCD} = -\frac{i}{2\sqrt{2}} \mathcal{G}_4.
$$

The next step is to provide an analytic framework to solve the equation
\begin{equation}
\sdiv_{2s} \phi = 0.
\end{equation}
This is done in \cite[Proposition 4.6]{MR3238530}.
\begin{proposition}[Andersson, B\"ackdahl, Joudioux, \cite{MR3238530}]\label{proprepspingeneral}
Let $\delta$ be in $\mathbb{R}\setminus\mathbb{Z}$, $j>0$ integer, $\phi_{A\dots F}$ in $\ker \sdiv_{2s} \cap H^j_\delta(S_{2s})$. Then there exist a spinor field $\tilde{\phi}_{A\dots F}\in H^{j+2s-1}_{\delta+2s-1}(S_{2s})$ and a constant $C$ depending only on $\delta$ and $j$ such that
\begin{align*}
\phi_{A\dots F}&=(\mathcal{G}_{2s}\tilde{\phi})_{A\dots F},\\
\Vert\tilde{\phi}_{A\dots F}\Vert_{j+2s-1,\delta+2s-1}&\leq C \Vert\phi_{A\dots F}\Vert_{j,\delta}.
\end{align*}
\end{proposition}
\begin{remark} The mechanism of the construction of a potential is explained in Appendix \ref{app:meca}.
\end{remark}

Finally, though it is a bit outside the scope of this note, we summarize here the relation between the zero rest-mass field, its Hertz potential, and their respective initial data, in the context of the Cauchy problem, as done in \cite[Theorem 7.6]{MR3238530}.
\begin{theorem}[Andersson, B\"ackdahl, Joudioux, \cite{MR3238530}]\label{th:repthm} Let $s$ be in $\tfrac{1}{2}\mathbb{N}$, $\delta$ be in $\R\setminus\mathbb{Z}$ and $j\geq2$ an integer. We consider $\phi_{A\dots F}$ in $H^j_\delta(S_{2s})$ satisfying the constraint equation $D^{AB}\phi_{A\dots F} =0$.
Then there exists a spinor field $\tilde{\phi}_{A\dots F}$, solving the equation
$$
\phi_{A\dots F} = (\mathcal{G}_{2s}\tilde{\phi})_{A\dots F}
$$
and satisfying the estimates
\begin{equation*}
\Vert\tilde{\phi}_{A\dots F}\Vert_{j+2s-1,\delta+2s-1}\leq C \Vert\phi_{A\dots F}\Vert_{j,\delta}.
\end{equation*}
Furthermore, the unique solution of the Cauchy problem for massless fields \eqref{cauchyprobmassless} with the initial datum $\phi_{A\dots F}$ is given by
$$
\phi_{A\dots F} = \nabla_{AA'}\dots \nabla_{FF'}\widetilde{\chi}^{A'\dots F'},
$$
where the spinor field $\chi_{A\dots F}$, defined by
$$
\chi_{A\dots F}= \tau_{AA'}\cdots \tau_{FF'}\widetilde\chi^{A'\dots F'},
$$ satisfies the Cauchy problem for the wave equation with initial data $(0,\sqrt{2}\tilde{\phi}_{A\dots F})$.
\end{theorem}

\section{Gluing for constraints for higher spin fields}\label{sec:glue}

This section contains the actual new mathematical contents of this note. It provides, on $\mathbb{R}^3$, a generalization to higher spin fields of \cite{2017arXiv170100486B}, in the context of weighted Sobolev spaces, and provides a detailed analysis of the possible asymptotic behaviour as $r$ goes to $\infty$. It must be noted that the approach developed in \cite{MR3238530} to prove the existence of initial data for the Hertz potential is based on basic elliptic theory, while the approach developed in \cite{2017arXiv170100486B} relies on a representation formula (from application of the Poincar\'e lemma) for potentials.

If $\Omega$ is an open set, consider $\Omega_\epsilon$ the $\epsilon$-neighbourhood of $\Omega$, that is to say the set
$$
\Omega_\epsilon =  \{x \in \mathbb{R}^3 |\exists y \in \Omega, x \in B(y, \epsilon) \}.
$$
We also define the neighbourhood of $\Omega$:
$$
\tilde{\Omega}_\epsilon =  \{x \in \mathbb{R}^3 |\exists y \in \Omega_{\epsilon}, x \in B(y, <x>\epsilon) \},
$$
and the neighbourhood of $\Omega_{\epsilon}$:
$$
\tilde{\tilde{\Omega}}_\epsilon =  \{x \in \mathbb{R}^3 |\exists y \in \tilde{\Omega}_{\epsilon}, x \in B(y, <x>\epsilon) \}
$$
We then provide two results of gluing. The first result holds for a small gluing region, $\Omega_{\epsilon}\setminus \Omega$, and contains a loss of control in the asymptotic region of the gluing while the second holds for a large gluing neighbourhood, and contains no loss of control in the asymptotic region of the gluing.
\begin{proposition}\label{prop1} Let $\delta$ be in $\mathbb{R}\setminus \mathbb{Z}$, $k$ be an integer, $\epsilon$ in $\mathbb{R}^+$, and $\Omega$ an (non-empty) open set of $\mathbb{R}^3$. Consider the solution to the constraint equation for spin-s fields
  $$
\sdiv_{2s} \phi = 0 \text{ with } \phi \in H^k_{\delta}(S_{2s}).
  $$
  There exists a spin-s field $\tilde{\phi}$ in $H^{k}_{\delta+k+2s-1}(S_{2s})$ satisfying
  \begin{itemize}
    \item $\tilde{\phi}$ satisfies the constraint equation
    $$
\sdiv_{2s}\tilde{\phi} = 0;
    $$
    \item on $\Omega$, $\phi = \tilde \phi$;
    \item on the complement of ${\Omega}_{\epsilon}$, $\tilde \phi = 0$.
  \end{itemize}
Furthermore, if $k\geq 2$, and if $\Omega$ is unbounded, for any $\ell \leq k-2$, then $D^\ell\tilde \phi$ has the following asymptotic behaviour in  $\Omega_{\epsilon}\setminus \Omega$:
$$
D^\ell\tilde{\phi}(x) = \mathcal{O}(|x|^{\delta+k+2s-1-\ell}) \text{ as } |x| \rightarrow +\infty, \text{ for } x \in {\Omega}_{\epsilon}\setminus \Omega.
$$
\end{proposition}
\begin{proof} By convolution, it is possible to construct a smooth function $\chi_{\Omega}$ in $W^{k+2s+1,\infty}(\mathbb{R}^3)$ such that $\chi_{\Omega}$ is equal to 1 in a neighbourhood of $\Omega$, and vanishes in the complement $\Omega_{\epsilon}$.

   Since $\phi$ belongs to the kernel of $\sdiv_{2s}$, and to the Sobolev space $H^k_{\delta}(S_{2s})$, by Proposition \ref{proprepspingeneral}, there exists a preimage $\psi$ of $\phi$, by the operator $\mathcal{G}_{2s}$, in $H^{k+2s-1}_{\delta+2s-1}(S_{2s})$, such that
  $$
\phi = \mathcal{G}_{2s} \psi.
  $$
The field defined by
$$
\tilde{\phi} = \mathcal{G}_{2s}\left(\chi_{\Omega} \psi  \right).
$$
satisfies the gluing properties.

Furthermore, noticing that, for $p \leq k+2s-1 $
\begin{eqnarray*} \label{eq:sum}
<x>^{p-k-2s+1}\vert D^p \left(\tilde{\chi} \phi \right) \vert  & \leq&  C  \left( \sum_{ 0 \leq l\leq p} <x>^{p-k-2s+1}  \vert D^{p-l}{\tilde{\chi}} D^l\phi\vert \right)\\
 & \leq&  C |\Vert \tilde{\chi} \Vert_{\infty, k+2s-1} \left( \sum_{ 0 \leq l\leq p}  \vert <x>^{(p-l)-k-2s+1 +l}  D^l\phi\vert \right)\\
  & \leq&  C |\Vert \tilde{\chi} \Vert_{\infty, k+2s-1} \left( \sum_{ 0 \leq l\leq p}  \vert  <x>^l D^l\phi\vert \right),
\end{eqnarray*}
since $p-l$ is at most $k+2s-1$, one obtains, that
there exists a constant $C$
  $$
\Vert \chi \psi \Vert_{k+2s-1, \delta + k+2s1}\leq C \Vert\chi \Vert _{W^{k+2s-1,\infty}} \cdot \Vert \psi \Vert_{k+2s-1, \delta},
  $$
Hence, $\mathcal{G}_{2s}(\chi \tilde\psi)$ belongs to $H^{k}_{\delta+k+2s+1}(S_{2s})$.

Furthermore, for $k\geq 2$, and $\Omega$ unbounded, using the standard weighted Sobolev estimates (see \cite[Theorem 1.2, (iv)]{Bartnik:1986dq}), one obtains, since $\tilde\phi \in H^{k}_{\delta+k+2s-1}(S_{2s})$, for any $\ell\leq k-2$:
$$
 D^\ell \tilde{\phi}(x) = \mathcal{O}(|x|^{\delta+k+2s-1-\ell}) \text{ as } |x| \rightarrow +\infty \text{ for } x \in \tilde{\Omega}_{\epsilon}\setminus \Omega.
$$
\end{proof}

It is possible to lose no decay in the gluing region if one increases its size:
\begin{proposition}\label{prop2} Let $\delta$ be in $\mathbb{R}\setminus \mathbb{Z}$, $k$ be an integer, $\epsilon$ in $\mathbb{R}^+$, and $\Omega$ an (non-empty) open set of $\mathbb{R}^3$. Consider the solution to the constraint equation for spin-s fields
  $$
\sdiv_{2s} \phi = 0 \text{ with } \phi \in H^k_{\delta}(S_{2s}).
  $$
  There exists a spin-s field $\tilde{\phi}$ in $H^k_{\delta}(S_{2s})$ satisfying
  \begin{itemize}
    \item $\tilde{\phi}$ satisfies the constraint equation
    $$
\sdiv_{2s}\tilde{\phi} = 0;
    $$
    \item on $\Omega$, $\phi = \tilde \phi$;
    \item on the complement of $\tilde{\tilde{\Omega}}_{\epsilon}$, $\tilde \phi = 0$.
  \end{itemize}
Furthermore, if $k\geq 2$, and if $\Omega$ is unbounded, for any $\ell \leq k-2$, then $D^\ell\tilde \phi$ has the same decay rate as $D^\ell\phi$ in the overlap region $\tilde{\Omega}_{\epsilon}\setminus \Omega$:
$$
D^\ell\tilde{\phi}(x) = \mathcal{O}(|x|^{\delta-\ell}) \text{ as } |x| \rightarrow +\infty, \text{ for } x \in \tilde{\tilde{\Omega}}_{\epsilon}\setminus \Omega.
$$
\end{proposition}
\begin{proof} The proof runs exactly as the previous proof. The main difference lies in the choice of the cut-off function to perform the mollification. Consider $\psi$ a pre-image of $\phi$ by $\mathcal{G}_{2s}$ in $H^{k+2s-1}_{\delta+2s-1}(S_{2s})$.

Let $\rho$ be a non negative smooth function, bounded by $1$, and with support in the ball of center the origin and radius $\epsilon$. Assume that
 $$
\int_{\mathbb{R}^3} \rho(x) dx =1.
$$
 Let $\chi_{\tilde{\Omega}_{\epsilon}}$ be the characteristic function of the set
 $$
\tilde{\Omega}_{\epsilon} =  \cup _{x\in\Omega} B(x, <x>\epsilon).
 $$
 Define the convolution, for all $x$ in $\mathbb{R}^3$,
$$
\tilde \chi (x)= \int_{\mathbb{R}^3} \chi_{\tilde{\Omega}_{\epsilon}} (y) \chi\left(\dfrac{x-y}{<x>}\right) \dfrac{dy}{\left(<x>\epsilon\right)^3} = \int_{\mathbb{R}^3} \chi_{\tilde{\Omega}_{\epsilon}} (x- <x>\epsilon u )\rho(u) du .
$$
By construction, the support of $\tilde{\chi}$ is contained in
$$
\{x \in\mathbb{R}^3 | B(x, <x>\epsilon) \cap \tilde{\Omega}_{\epsilon}\neq \emptyset \} \subset \cup _{x\in\tilde{\Omega}_{\epsilon} } B(x, <x>\epsilon).
$$
Furthermore, using $\int_{\mathbb{R}^3} \rho(u) du = 1$, one obtains
$$
1- \tilde \chi (x) = \int_{\mathbb{R}^3} (1 - \chi_{\tilde{\Omega}_{\epsilon}}) (x- <x>\epsilon u )\rho(u) du =  \int_{\mathbb{R}^3}  \chi_{\tilde{\Omega}^c_{\epsilon}} (x- <x>\epsilon u )\rho(u) du,
$$
where $\tilde{\Omega}^c_{\epsilon}$ is the complement of $\tilde{\Omega}_{\epsilon}$ in $\mathbb{R}^3$. Hence $1-\tilde{\chi}$
 vanishes on
$$
\{x \in\Omega | B(x, <x>\epsilon) \cap \tilde{\Omega}^c_{\epsilon}= \emptyset \} = \tilde{\tilde{\Omega}}_{\epsilon}.
$$
Finally, we notice that, following the same argument as in Equation \eqref{eq:sum}
\begin{eqnarray*}
\vert D^k \tilde{\chi} (x)\vert\leq  \dfrac{C(\epsilon)}{<x>^k}\sum_{ 0 \leq p \leq k}\Vert D^p \rho \Vert_{\infty},
\end{eqnarray*}
where $C(\epsilon)$ is a constant depending on $\epsilon$. Hence, by the same argument as in Equation \eqref{eq:sum}, one obtains that
$$
\Vert\tilde{\chi} \psi  \Vert_{k+2s-1, \delta} \leq C \Vert \psi  \Vert_{k+2s-1, \delta}.
$$

The rest of the proof goes as the proof of Proposition \ref{prop1}.

\end{proof}

\appendix
\section{Mechanism of the construction of a potential}\label{app:meca}

We explain here the mechanism that leads to the construction of a potential for solutions to the geometric constraints for higher spin fields. This calculation is based on the assumption that the Laplacian is bijective (this would be the case for spinors with rapidly decaying coefficients, or for fields in $H^k_{0}(S_{2s})$, with $\delta$ in $(0,1)$). To handle the case when the Laplacian is not surjective or injective, further insights on the structure of the kernel of $\sdiv$ is required, and more specifically, a Helmholtz type decomposition involving $\mathcal{G}$, $\stwist$ and the Laplacian (see \cite[Lemmata 4.1 and 4.4]{MR3238530}) is required.

We deal first with the case of the Maxwell equations in the formalism of forms: consider a 1-form $E$ on the Euclidean $\mathbb{R}^3$ satisfying the Maxwell constraints:
$$
d^\star E = 0.
$$
Let $\hat{E}$ be a preimage by the Laplacian of $E$:
\begin{equation}\label{eq:a1}
E = \Delta \hat{E} =  (dd^{\star} + d^{\star} d) \hat{E},
\end{equation}
where $d^{\star} = (-1) \star d \star$, $\star$ being the Hodge dual. Applying the operator $d^{\star}$ to $E$ leads to
$$
d^{\star} dd^{\star}\hat{ E} =  (dd^{\star} + d^{\star} d)  d^{\star} \hat{E} = 0,
$$
that is to say that $d^{\star} \hat{E}$ belongs to the kernel of the Laplacian, and vanishes. Hence, $E$ can be written:
$$
E =  d^{\star} d  \hat{E} =  \star d \star d \hat{E} = \star d \tilde{E},
$$
where $\tilde{E} =\star d \hat{E}$ is the researched $1$-form.

The same approach can in fact be followed to construct a potential for solutions of the constraint equation for higher spin fields
$$
\sdiv_{2s} \phi = 0.
$$
To this end, one needs to find a relation between the operators of the elliptic complex of Lemma  \ref{lem:ellipticcomplex}, and powers of the Laplacian. These relations are stated explicitly in \cite[Equations (2.2a) and (2.2b)]{MR3238530}. We state here these equations for a spinor of valence $2k$ and $2k+1$ respectively:
\begin{subequations}
\begin{align}
(\Delta^{k}_{2k}\phi)_{A_1\dots A_{2k}}={}&(\stwist_{2k-2}\mathcal{F}_{2k-2}\sdiv_{2k}\phi)_{A_1\dots A_{2k}}-(-2)^{1-k}(\mathcal{G}_{2k}\scurl_{2k}\phi)_{A_1\dots A_{2k}},\label{LaplacianAsGeven}\\
 (\Delta^{k}_{2k+1}\phi)_{A_{1}\dots A_{2k+1}}={}&
(\stwist_{2k-1}\mathcal{F}_{2k-1}\sdiv_{2k+1}\phi)_{A_1\dots A_{2k+1}}
+(-2)^{-k}(\mathcal{G}_{2k+1}\phi)_{A_1\dots A_{2k+1}} \label{LaplacianAsGodd},
\end{align}
\end{subequations}
where the operators $\mathcal{F}_{2s}$ for $s\in \tfrac{1}{2}\mathbb{N}_0$ are defined via
\begin{align*}
(\mathcal{F}_{2s}\phi)_{A_1\dots A_{2s}}={}&
2^{-2s}\sum_{n=0}^{\lfloor s\rfloor}\sum_{m=0}^{\lfloor s\rfloor-n}\binom{2s+2}{2n+2m+2}(-2)^{n}\nonumber\\
&\quad\times \underbrace{D_{(A_1}{}^{B_1}\cdots D_{A_{2n}}{}^{B_{2n}}}_{2n} (\Delta^{\lfloor s\rfloor-n}_{2s}\phi)_{A_{2n+1}\dots A_{2s})B_1\dots B_{2n}}.
\end{align*}

Once these formulae are recovered, the same argument as for the Maxwell constraint equations can be followed. To simplify the notations, we drop the labels in $k$ indicating the valence of the spinor on which the operator is applied to. We perform the calculation in the integer spin case. Let $\phi$ be a solution to the equation
$$
\sdiv \phi = 0.
$$
Let $\hat{\phi}$ be a preimage by $\Delta^k$:
$$
\phi = \Delta^k \hat{\phi} = \stwist \mathcal{F} \sdiv \hat{\phi} +\frac12 \mathcal{G} \scurl \hat{\phi}.
$$
We want to prove that
$$
\stwist \mathcal{F} \sdiv \hat{\phi} = 0.
$$
To this end, we notice that, since $\Delta^k$ and $\sdiv$ commute,
$$
\sdiv \Delta ^k \hat{\phi} = \Delta ^k \sdiv \hat{\phi} = 0.
$$
Hence, $\sdiv \hat{\phi}$ vanishes, and
$$
\phi = \frac12 \mathcal{G} \scurl \hat{\phi}.
$$
Our potential $\tilde{\phi}$ is then given by
$$
\tilde{\phi} = \frac12\scurl \hat{\phi}.
$$
\section*{Acknowledgment}
I wish to thank R. Beig and P. Chru\'sciel who suggested to type this note, and for their comments, and T. B\"ackdahl and L. Andersson who allowed me to use previous material. I also wish to thank L. Andersson, P. Chru\'sciel, D. H\"afner and J.-P. Nicolas for their continuous support.

The author is supported in part by the ANR grant AARG, "Asymptotic Analysis in General Relativity" (ANR-12-BS01-012-01).

\printbibliography

\end{document}